\documentclass[letterpaper, 10 pt, conference,openany]{ieeeconf}  
\IEEEoverridecommandlockouts                              
\usepackage[utf8]{inputenc}
\usepackage[T1]{fontenc}
\usepackage[mathscr]{eucal}
\usepackage{amssymb}
\usepackage{amsmath}
\usepackage{mathtools}
\usepackage{stackengine}
\usepackage{systeme,mathtools}
\usepackage{tikz} 
\usepackage{pgfplots}
\pgfplotsset{compat=newest}
\pgfplotsset{plot coordinates/math parser=false}
\newlength\figureheight
\newlength\figurewidth
\usepackage{cite}
\usepackage{setspace}
\usepackage{algorithm}
\usepackage{algpseudocode}

\newcounter{teocount}
\newcounter{propcount}

\newcounter{remcount}
\newcounter{defcount}

\newtheorem{remm}[remcount]{Remark}

\newtheorem{definition}[defcount]{Definition}
\newtheorem{proposition}[propcount]{Proposition}
\newtheorem{theorem}[teocount]{Theorem}

\newenvironment{remark}{\begin{remm}\rm }{\hfill \hspace*{1pt} \hfill
$\star$\end{remm}}

\usepackage{multicol}
 \usepackage{listings}
\usepackage[numbered]{matlab-prettifier}
\usepackage{tikz}
 
\usepackage{amsmath}
\usepackage{arydshln}
\usepackage{tikz}
\usetikzlibrary{positioning}

\tikzstyle{block} = [draw, fill=white, rectangle, 
    minimum height=3em, minimum width=6em]
\tikzstyle{sum} = [draw, fill=white, circle, node distance=1cm]
\tikzstyle{input} = [coordinate]
\tikzstyle{output} = [coordinate]
\tikzstyle{pinstyle} = [pin edge={to-,thin,black}]

\newcounter{example}[section]

\title{\huge \bf
Regional Stability Analysis of Transitional Fluid Flows
}

\author{ Leonardo F. Toso, Ross Drummond and Stephen R. Duncan
\thanks{}
\thanks{
        {\small Leonardo F. Toso, Ross Drummond and Stephen R. Duncan are with the Department of Engineering Science, University of Oxford, Oxford OX1 3PJ, United Kingdom. Email: \texttt{\{leonardo.toso, ross.drummond, stephen.duncan\}@eng.ox.ac.uk}.}}%
}

\begin{document}

\maketitle
\thispagestyle{empty}
\pagestyle{empty}

\begin{abstract}

A method to bound the maximum energy perturbation for which regional stability of transitional fluid flow models can be guaranteed is introduced. The proposed method exploits the fact that the fluid model's nonlinearities are both lossless and locally bounded and uses the axes lengths of the ellipsoids for the trajectory set containment as variables in the stability conditions. Compared to existing approaches, the proposed method leads to an average increase in the maximum allowable energy perturbation of   $\approx 29\%$ for the Waleffe-Kim-Hamilton (WKH) shear flow model and of $\approx 38\%$ for the 9-state reduced model of Couette flow.


\end{abstract}

\section{Introduction}
\label{sec:introduction}

Determining the stability properties of fluid flows remains a longstanding open problem tracing its roots back to Osborne Reynolds' 1883 experiments on the transition to turbulence in pipe flow  \cite{Reynolds}. The issues faced in predicting fluid stability are widely believed to be a result of the complex nature of the Navier-Stokes equations, which has forced practitioners to  either solve these equations numerically using computational fluid dynamics (CFD) or adapt experimental results to predict a fluid's response. Both of these methods have their limitations; CFD simulations are computationally demanding and require expertise to run, while experimental results can be expensive and are also typically designed for  demonstrative situations that may not generalise well to the flows found in practice. As a result, the design of many fluid-based technologies remain based upon significant experimental know-how and large computation power, an expensive and non-scalable situation. 

The limitations of CFD simulations and experimental characterisations have led to the development of several reduced-order fluid models for particular flows, which been shown to, at least qualitatively, give an indication of flow stability, while being significantly simpler to resolve than the general Navier-Stokes equations. Examples include the 4-state Waleffe-Kim-Hamilton (WKH) model \cite{Waleffe,Waleffe_Fabian} for shear flow and the 9-state reduced-order model \cite{Moehlis} for Couette flow bounded by two plates. The apparent success of these reduced-order fluid models in capturing the main features of the flow, has prompted research into their stability analysis, with the long term goal of this line of research being to generalise of the lessons learned from these simpler systems to develop scalable and non-conservative techniques for the analysis of more complex fluid models.

However, the stability analysis of even these reduced order models can still be challenging, owing to their nonlinear dynamics and non-normal state-transition matrices \cite{trefethen1, trefethen2}. Progress has been made, notably in \cite{Seiler_1, Seiler_2, Dennice}, where it was observed that the nonlinear model dynamics exhibit some structure that can be exploited. In particular, these results observed that the model dynamics could be understood in terms of the feedback interconnection of a linear system with an energy persevering or lossless nonlinear gain, as shown in Figure \ref{fig:lure_decomposition}, allowing the powerful and scalable techniques of passive systems theory (\cite[Chapter VI]{Khalil}) to be applied. However, stability certificates based upon passive systems theory have been found to be conservative and the maximum energy flow perturbation for which stability can be certified is significantly lower than that predicted by simulation (see Section \ref{sec:numerical_examples}). Reducing this conservatism will be necessary if these methods are to be deployed in practical applications involving more complex fluid flows. 

\emph{Contribution}: Motivated by the passive systems theory results \cite{Seiler_1, Seiler_2, Dennice}, this paper extends the approach by introducing an algorithm that allows the axes lengths of the ellipsoids bounding the state trajectories to be defined as matrix variables to be optimised over. To evaluate the performance of the method, the obtained stability conditions were  applied to both the 4-state WKH model for shear flow and the 9-state model of Couette flow and demonstrated a  reduction in conservatism compared to \cite{Seiler_1, Seiler_2, Dennice} (see Section~\ref{sec:numerical_examples}). These results indicate the potential of the proposed approach for analysing the stability of more complex flows where the increase in system dimension and/or complexity would makes it computationally impracticable to apply existing non-conservative methods, such as nonlinear direct-adjoint looping (DAL) \cite{Kerswell} or sum-of-squares programming \cite{Goulart,Valmorbida2}.

\emph{Paper structure:} The paper is structured as follows.  Section~\ref{sec:transitional_flow_models} introduces the two transitional fluid flow models and Section~\ref{sec:feebback_interconnection_of_the_models} consideres these models as the feedback interconnection of a linear system with a lossless nonlinearity, which allows  passive systems theory to be applied to their analysis. Sections~\ref{sec:local_quadratic_bounds} and \ref{sec:computation_of_ellipsiods_for_the_set_containmentt} formulate bounds for the nonlinear flow interactions by exploiting the fact that the system's state trajectories can be bounded within a ellipsoidal region. By exploiting these local quadratic bounds, a method to certify the regional stability of these fluid models is developed in Section~\ref{sec:regional_stability_analysis}, with the main result presented in Theorem \ref{theorem:local_state_bounds}. Numerical results  estimating the region of attraction of the Waleffe-Kim-Hamilton (WKH) model and the 9-state reduced-order model of Couette flow are described in Section~\ref{sec:numerical_examples} to illustrate the potential of the approach.

\emph{Notation:} The identity matrix of dimension $n$ is $I_n$ and the the matrix of zeros of dimension $n \times m$ is $0_{n \times m}$. If a matrix $A$ is positive definite then $A \in \mathcal{S}^n_{\succ 0}$ and if it is a diagonal matrix with positive diagonal elements then $A \in \mathbb{D}^n_{++}$. The $\rho$-level sets of a function $V(x)$ are defined as $\mathcal{E}(V,\rho):= \{x: V(x) = \rho\}$.


\section{Transitional Fluid Flow Models}
\label{sec:transitional_flow_models}

The transitional fluid flow models considered in this paper are derived from direct numerical simulations (DNS) of plane Couette flow \cite{Waleffe_Fabian}. Both the 4-state and 9-state reduced-order models are described by ordinary differential equations (ODEs) and are derived from simplifications of the Navier-Stokes equations.

\subsection{General form}


Both of the considered transitional fluid flows models can be expressed in the general form
\begin{equation}
    \dot{x}(t)=Ax(t) + \phi(x)
    \label{nonlinear_system}
\end{equation}
where $x \in \mathbb{R}^n$ is the system's state, $A \in \mathbb{R}^{n\times n}$ is the Hurwitz state transition matrix depending on the Reynolds number ($Re$), and $\phi(x)=\mathcal{F}x \in \mathbb{R}^n \to\mathbb{R}^n $ describes the nonlinear interactions of the fluid flow. The nonlinearity $\phi(\cdot)$ can be expressed in terms of a quadratic form
\begin{equation}
    \phi(x)=\begin{bmatrix}
                x^\top S_1 x\\
                \vdots\\
                x^\top S_n x
    \end{bmatrix}
\end{equation}
with $S_1,\ldots,S_n$ being symmetric matrices.

\subsection{Waleffe-Kim-Hamilton (WKH) shear flow model}

With the WKH model, the behaviour of a shear flow bounded by two plates, one moving and the other stationary, is described by
\begin{align}
\small
\begin{bmatrix} \dot{x}_1 \\ \dot{x}_2 \\ \dot{x}_3 \\ \dot{x}_4 \end{bmatrix} 
=
\frac{1}{Re}\begin{bmatrix} -\lambda  & 1 & 0 & 0 \\ 0 & -\mu & 0 & 0\\ 0 & 0 & -\nu & 0 \\ 0 & 0 & 0& -\sigma \end{bmatrix}
\begin{bmatrix} {x_1} \\ {x_2} \\ {x_3} \\ {x_4} \end{bmatrix}
+
\begin{bmatrix} x_2x_4-\gamma x_3^2 \ \\ \delta x_3^2 \\ \gamma x_3x_1 -\delta x_3x_2 \\ -x_2x_1 \end{bmatrix}
    \label{eq:WKH_model}
\end{align}
\noindent where $x_1$, $x_2$, $x_3$ and $x_4$ represent the amplitude of the stream-wise velocity, rolls (that consists mostly of vertical velocity), inflectional streak instability and mean shear, respectively. The positive constants $\lambda$, $\mu$, $\sigma $ and $\nu$ concern the viscous decay rates, whereas $\gamma$ and $\delta$ are positive nonlinear interaction coefficients.

\subsection{9-state reduced model of Couette flow}

The 9-state reduced-order model of Couette flow is a low-dimensional model for turbulent shear flows generalising the eight-mode model of \cite{Fabian_Waleffe} to capture variations in the main fluid velocity profile during the transition from laminar to turbulent states. The nine ordinary differential equations of the model are detailed in Appendix A and are obtained by applying a Galerkin projection \cite{Holmes} on the mode profiles over the spatial domain $0\leq b\leq L_b$, $-1\leq c\leq 1 $ and $0\leq d \leq L_d$, where $a,\,b$ and $c$ relate to the downstream, shear and spanwise spatial directions, respectively.

\section{Feedback Interconnection of the models}
\label{sec:feebback_interconnection_of_the_models}

Accounting for the nonlinear terms $\phi(\cdot)$ in \eqref{nonlinear_system} is the main source of difficulty in the stability analysis of these models.  However, for the transitional flow models considered here, these nonlinear terms exhibit properties that can be exploited. Specifically, well-established theory on the incompressible Navier-Stokes equations \cite{Schmid} (highlighted through the Leray formulation \cite{Leray}) means that for many wall-bounded transitional fluid flow models, including \cite{Waleffe} and \cite{Moehlis}, the nonlinearity $\phi(x)$ is \emph{memoryless}, meaning that the mapping $\mathcal{F}: \mathbb{R}^n \to \mathbb{R}^n$ does not vary with time, satisfies $\phi(0)=0$, and is \emph{lossless}.

\begin{definition}
(Losslessness \cite{Khalil}) A nonlinear real function $\phi(x)=\mathcal{F}x$, with $\mathcal{F}: \mathbb{R}^n \to \mathbb{R}^n$, is said to be lossless if the following condition is verified:
\begin{equation}
    x^\top \phi(x)=0, \quad \forall x \in \mathbb{R}^n.
\end{equation}
\label{def:losslessness}
\end{definition}
Losslessness of $\phi$ can also be encoded in a matrix form
\begin{equation}
     \eta^\top \underbrace{\begin{bmatrix} 0_{n\times n} & I_{n}\\
     I_{n} & 0_{n\times n}
     \end{bmatrix}}_{F_0} \eta = 0, \quad  \forall x \in \mathbb{R}^n,
    \label{eq:local_quadratic_constraint_1}
\end{equation}
\noindent with $\eta^\top = [x^\top,\;\ \phi(x)^\top]$, since $x^\top \phi(x)=\phi^\top (x) x= 0$.

 By introducing the additional variable $v$,  the model dynamics \eqref{nonlinear_system} can be equivalently written as
\begin{subequations}
\begin{align}
    \dot{x}&=Ax + v,\\
    v&=\phi(x).
\end{align}
\end{subequations}

As illustrated in Figure \ref{fig:lure_decomposition}, this system can be understood as the feedback interconnection of a linear system with transfer function $(sI-A)^{-1}$ mapping $v \to x$ where $v$ is the output from mapping the state $x$ through the nonlinear, but  lossless, gain $\phi(\cdot)$. The losslessness property of $\phi(\cdot)$ means that the system's stability can be inferred using passive systems theory \cite{Khalil}, which directly exploits this feedback based perspective.

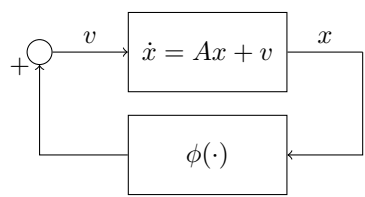
\begin{figure}[h!]
\centering
   \begin{tikzpicture}[auto, node distance=1cm]
    \node [sum, ] (sum) {};
    \node [block, right = 1cm of sum] (system) {$\dot{x}=Ax + v$};
    \node [output, right= 1cm of system] (output) {};
    \node [block, below = 0.3cm of system] (nonlinearity) {$\phi(\cdot)$};
    \draw [->] (sum) -- node {$v$} (system);
    \draw [-] (system) -- node {$x$} (output);
    \draw [->] (output) |- (nonlinearity);
    \draw [->] (nonlinearity) -| node[pos=0.99] {$+$} 
        node [near end] {} (sum);
\end{tikzpicture} 
\caption{\textit{Feedback interpretation of  \eqref{nonlinear_system} in terms of a linear system with a lossless nonlinearity.}}
\label{fig:lure_decomposition}
\end{figure}

\section{Local quadratic bounds for $\phi$}
\label{sec:local_quadratic_bounds}

Globally, when considering $x\in \mathbb{R}^n$, losslessness is one of the only properties satisfied by $\phi(x)$, however,  when considering a regional analysis with $x\in \mathcal{E} \subset \mathbb{R}^n$, the nonlinearity $\phi(x)$ can be locally bounded. The benefits of introducing local bounds for $\phi(\cdot)$  was identified in \cite{Seiler_2}, with the bounds obtained from the Cauchy-Schwartz Lemma allowing the conservatism of their stability certificates to be reduced compared to the earlier results of \cite{Seiler_1} and \cite{Dennice}. However, in general, bounds produced by the Cauchy-Schwartz Lemma are not tight, which suggests that there may be room to reduce this conservatism still further. Here, it is shown how by  inspecting the various terms of the nonlinear term directly and bounding them, additional quadratic bounds for products of the nonlinear terms can be specified.

\subsection{Local quadratic bounds}
For the WKH model of \eqref{eq:WKH_model}, assume that the states are bounded by
\begin{align}
    x_1^2 \leq {\gamma_{x_1,k}},~x_2^2 \leq {\gamma_{x_2,k}},~ x_3^2 \leq {\gamma_{x_3,k}},~ x_4^2 \leq {\gamma_{x_4,k}}
\end{align}
for $k = 1, 2, \, \dots, K$ and define
\begin{align}
    & \hat{\gamma}_{x_1,k} = {\lambda_k \gamma_{x_1,k}},\quad
     \hat{\gamma}_{x_2,k} = {\lambda_k\gamma_{x_2,k} }, \\
  & \hat{\gamma}_{x_3,k} = {\lambda_k\gamma_{x_3,k}},\quad
       \hat{\gamma}_{x_4,k} = \lambda_k \gamma_{x_4,k}.\nonumber
\end{align}

Following some manipulation, the WKH model's nonlinear terms can be, for instance, bounded by
\begin{subequations}
\small
\begin{align}
\lambda_1\phi_2(x)^2 &\leq x_3^2 {\hat{\gamma}_{x_3,1}}, \\
\lambda_2\phi_3(x)^2&\leq {\hat{\gamma}_{x_3,2}}(x_1 - x_2)^2,\\
\lambda_3\phi_4(x)^2 &\leq {\hat{\gamma}_{x_2,3}} x_1^2,\\
\lambda_4\phi_4(x)^2 &\leq {\hat{\gamma}_{x_1,4}} x_2^2,\\
\lambda_5\phi_3(x)^2 &\leq {\hat{\gamma}_{x_3,5}}(x_1^2+(x_1-x_2)^2 + (\delta v)^2),\\
\lambda_6\phi_1(x)\phi_4(x) &\leq \frac{1}{2} {\hat{\gamma}_{x_3,6}}(x_2+x_1)^2 + \frac{1}{2}\hat{\gamma}_{x_2,6}(x_4-x_1)^2,\\
\lambda_7\phi_2(x)\phi_4(x) &\leq \frac{1}{2} \hat{\gamma}_{x_3,7}(x_2-x_1)^2,\\
\lambda_8\phi_3(x)\phi_4(x) &\leq \frac{1}{2} \hat{\gamma}_{x_1,8}(x_3-x_2)^2 +\frac{1}{2} \hat{\gamma}_{x_2,8} (x_3+x_1)^2 .
\end{align}
\end{subequations}

Using the same approach, other bounds for $\phi$ can be generated and the non-linearities of the 9-state Couette flow model can also be similarly bounded. 



\subsection{Matrix Inequalities}

The local quadratic bounds of each models' nonlinearities (defined in the previous sub-section with $K$ being the total number of bounds obtained) can be expressed in a matrix form
\begin{equation}
    \eta^\top \underbrace{\begin{bmatrix} \hat{\Gamma}_k(\hat{\gamma}_{x_j,k}) & 0_{n\times n}\\
                      0_{n\times n}      & -\lambda_kM_k
      \end{bmatrix}}_{F_k(\hat{\Gamma}_{k},\,\lambda_k)}\eta \geq 0, ~\forall x \in \mathcal{E},
      \label{eq:local_quadratic_bounds_2}
\end{equation}

\noindent with $M_k$ being a symmetric  matrix corresponding to the $k^\text{th}$ bound, with $k = 1,\, \dots, \, K$. By incorporating these local bounds using  the \emph{S-procedure} \cite{Aizerman},  local information on $\phi(\cdot)$ can be included within the stability analysis of Theorem~\ref{theorem:local_state_bounds}, helping to reduce the conservatism.

As described in \cite[Lemma 1]{Seiler_1}, the nonlinearity $\phi(x)$ can also be bounded by
\begin{align}
    \eta^\top \underbrace{\begin{bmatrix}  S_i \Delta_i S_i & 0_{n\times n}\\
                       0_{n\times n}      & -\lambda_{K+i}e_i^Te_i
      \end{bmatrix}}_{M_i(\Delta_i,\,\lambda_{K+i})}\eta \geq 0, \forall x \in \mathcal{G}_{\xi}
      \label{eq:local_quadratic_bounds_3}
\end{align}
for all $i=1,\ldots,n$, where $\mathcal{G}_{\xi}=\{x \in \mathbb{R}^n: x(t)^\top G x(t) \leq \xi^2\}$ defining an ellipsoid over $\mathbb{R}^n$ and with $e_i \in \mathbb{R}^n$ being the $i^\text{th}$ standard basis vector. Contrasting with the formulation of \cite{Seiler_2}, in this paper, the matrices $\Delta_i=\lambda_i\xi^2G^{-1}$ are defined as matrix decision  variables in the optimisation problem of the stability conditions, instead of being fixed at each iteration.

\section{Ellipsoids for the set containment}
\label{sec:computation_of_ellipsiods_for_the_set_containmentt}

For the local quadratic bounds on the nonlinear terms to hold, the state trajectories must be constrained to the local region $x \in \mathcal{E}$ for all  initial conditions considered. The following proposition allows the ellipsoidal sets for this set containment to be posed in terms of linear matrix inequalities.

\begin{proposition}\label{prop:set_cont}
 Consider a Lyapunov function $V(x): \mathbb{R}^n \to \mathbb{R}_+ = x(t)^TPx(t)$ with $P \in \mathbb{S}_{\succ 0}^n$. For $k = 1, \,2, \, \dots,\,K$ with $K$ being the total number of bounds for $\phi(x)$ (Section \ref{sec:local_quadratic_bounds}). Define  ellipses $E_k(x) = x^T\Lambda_k^{-1} x = \sum^{n}_{j = 1}\gamma_{x_j,k}^{-1}{x_j}^2$ where $\gamma_{x_j,k} >0 $,  and the matrix $\hat{\Lambda}_k^{-1} = \frac{1}{\lambda_k}\Lambda_k^{-1}$.

 If  
\begin{subequations}\label{matrix_conditions}\begin{align} \label{level_sets}
 \begin{bmatrix}P & \bar{\lambda}_k^{1/2}I_n \\ \bar{\lambda}_k^{1/2}I_n & {{\hat{\Lambda}}_k}\end{bmatrix}
 \succ 0,
 \\
 \bar{\lambda}_k\geq \lambda_k >0, \label{lambda_dom}
\end{align}
\label{set_conditions}\end{subequations}
 then $\mathcal{E}(V,1) \subseteq \mathcal{E}(E_k,1)$ where $\mathcal{E}(V,1):=\{x \in \mathbb{R}^n: V(x)\leq 1\}$ and $\mathcal{E}(E_k,1):=\{x \in \mathbb{R}^n: x^\top\Lambda_k x\leq 1\}$.
\label{prep:local_quadratic_bounds}
\end{proposition}

\begin{proof}
From the Schur complement, \eqref{level_sets} is equivalent to
\begin{align}
    P- {\bar{\lambda}_k}^{1/2}I_n\hat{\Lambda}_k^{-1}I_n{\bar{\lambda}_k}^{1/2}
    =
    P- \frac{\bar{\lambda}_k}{ \lambda_k }{\Lambda_k^{-1}}\succ 0.\label{P_cond1}
\end{align}
Since $\bar{\lambda}_k \geq {\lambda}_k >0$, then \eqref{P_cond1} implies
$
    P- {\Lambda_k}^{-1}\succ 0.
$
Multiplying this matrix inequality on the left by $x^\top$ and on the right by $x$ gives
\begin{align}
    E_k(x)\leq V(x).
\end{align}
\noindent We then have the set containment $\mathcal{E}(V,1) \subseteq \mathcal{E}(E_k,1)$.
\end{proof}

\section{Regional stability of transitional flow models}
\label{sec:regional_stability_analysis}

Conditions to estimate the regional stability analysis can be formulated using the set containment of Proposition \ref{prop:set_cont}. In keeping with recent results, e.g. \cite{Seiler_2}, these conditions are posed by computing inner estimates of the maximum energy perturbation for which asymptotic stability can be guaranteed.

\begin{theorem}
Consider the system \eqref{nonlinear_system}. For  given $\epsilon>0$ and ${\bar{\lambda}} \in \mathbb{R}^{n+K}_{\geq 0}$, if there exists positive-definite matrices $P~\in~\mathbb{S}^{n}_{\succ 0}$,  ${\Delta}_i$ $ \in \mathbb{D}^n_{++}$, $\hat{\Gamma}_k$ $ \in \mathbb{D}^n_{++}$ and Lagrange multipliers $\zeta_0 \in \mathbb{R}$, $\lambda_i \in \mathbb{R}^{n+K}_{\geq 0}$ that solves

\begin{subequations}\begin{align}
\beta^{*}= \frac{1}{(R^*)^2}=
 \underset{P,\,\beta,\, \Delta,\,\hat{\Lambda}(\hat{\gamma}_{x_j,k}),\,\lambda,\,\zeta_0}{\text{min}} \quad \beta = \frac{1}{R^2}
\end{align}
subject to 

\begin{align}
   & \begin{bmatrix}
  A^\top P + PA  & P\\
  P     &        0_{n\times n}
  \end{bmatrix} + \zeta_0F_0 + \sum_{k=1}^{K}F_k(\hat{\Gamma}_k,\lambda_k)\nonumber\\
  &\qquad \qquad + \sum_{i=1}^{n} M_i(\Delta_{i},\lambda_{K+i}) \preceq -\begin{bmatrix}
  \epsilon I_{n\times n} & 0_{n\times n}\\0_{n\times n} & 0_{n\times n}
  \end{bmatrix},\label{eq:theorem_local_state_bound_constraint_1}\\
  &\begin{bmatrix}P & \bar{\lambda}_k^{1/2}I_n \\ \bar{\lambda}_k^{1/2}I_n & \hat{\Lambda}_k(\hat{\gamma}_{x_j,k})\end{bmatrix}
 \succ 0, \;\ \forall k=1,\,\ldots,\,K, \label{eq:theorem_local_state_bound_constraint_2}
 \\
  &\begin{bmatrix}P & \bar{\lambda}_{K+i}^{1/2}I_n \\ \bar{\lambda}_{K+i}^{1/2}I_n & {\Delta}_i\end{bmatrix}
 \succ 0, \;\ \forall i=1,\,\ldots,\,n, \label{eq:theorem_local_state_bound_constraint_3}
 \\
 &0 \prec P \preceq \beta I_n \label{eq:theorem_local_state_bound_constraint_5}\\
&\bar{\lambda}_\ell\geq \lambda_\ell > 0,\;\ \forall \ell=1,\,\ldots,\,K+n\label{eq:theorem_local_state_bound_constraint_6}\\
&\hat{\Lambda}_k(\hat{\gamma}_{x_j,k})\succeq 0, \;\ \forall k=1,\,\ldots,\,K  ,
\label{eq:theorem_local_state_bound_constraint_7}\\
& \Delta_i\succeq 0, \;\ \forall i=1,\,\ldots,\,n , \label{eq:theorem_local_state_bound_constraint_8}
\end{align}
\end{subequations}
\noindent then the system is asymptotically stable for all initial conditions $x(0) \in \mathcal{R}:=\{x(0) \in \mathbb{R}^n: x(0)^\top x(0) \leq ({R^*)}^2\}$ with the trajectories satisfying the set containment  $x \in \mathcal{R} \subseteq \mathcal{E}(V,1) \subseteq \mathcal{E}(E_k,1)$. 

\label{theorem:local_state_bounds}

\begin{proof}
With the Lyapunov function $V(x(t)) = x(t)^\top Px(t)$, condition \eqref{eq:theorem_local_state_bound_constraint_1} implies $\dot{V}(x(t)) <0~\forall x \in  \mathcal{E}(E_k,1)$. It is then required to show that the state trajectories remain within $\mathcal{E}(E_k,1)$ at all times. For this, it is noted that \eqref{eq:theorem_local_state_bound_constraint_2} and \eqref{eq:theorem_local_state_bound_constraint_6} imply via Proposition \ref{prop:set_cont} that   $\mathcal{E}(V,1) \subseteq \mathcal{E}(E_k,1)$. From \eqref{eq:theorem_local_state_bound_constraint_5}, then $0 \leq V(x) \leq \beta x(t)^\top x(t)$ when $x(0)^\top x(0) \leq R^2$, then $V(x(0)) \leq 1$. Condition \eqref{eq:theorem_local_state_bound_constraint_1} means that the sublevel sets of $\mathcal{E}(V,1) $ are positive invariant, giving asymptotic stability.

\end{proof}
\end{theorem}

\begin{remark}
The formulation of Theorem \ref{theorem:local_state_bounds} allows the axes lengths of the ellipsoids for the set containment (incorporated through the matrices $\hat{\Gamma}$ and $\hat{\Delta}$) to be decision variables in the problem. This formulation contrasts with \cite[Algorithm A]{Seiler_2} where the ellipses are fixed at each iteration. However, in order to convexify the problem, the upper bound $\bar{\lambda}$ for the Lagrange multipliers  ${\lambda}$ have to be fixed when variable axes lengths are used, which motivates the following convexification. 
\end{remark}

\begin{algorithm}[t]
	\caption{Compute maximum energy perturbation $R^{*}$}
	\begin{algorithmic}[1]
	\State \textbf{Set-up:}  Obtain $A$ and tolerances $\epsilon$, $\epsilon_{\beta}$. Set $m = 0$
    \State \textbf{Initialisation:}$\left(P^{\{0\}},\lambda^{\{0\}},\beta^{*\{0\}}\right)$ $\leftarrow$ Solve Theorem \ref{algorithm_1}
    with  \eqref{eq:theorem_local_state_bound_constraint_6} replaced by \eqref{eq:lambda_mod}
     \State $\bar{\lambda}^{\{0\}}=\frac{\lambda^{\{0\}}}{\beta^{*\{0\}}}\max\left(\text{eig}(P^{\{0\}})\right)$
     \State \textbf{Define:} Directions $v^{\{j\}}\in \mathbb{R}^{n+K}$ for $j = 1, \,\dots, \, n_v$
     \State Set $\beta^{*\{0\}} = 0$ and $\omega^{\{0\}} = \bar{\lambda}^{\{0\}}$
 \While{$\beta^{*\{m+1\}} - \beta^{*\{m\}}\geq \epsilon_{\beta}$ }
    \For{$j = 1,\, \ldots, \, n_v$}
               \State Solve Theorem \ref{algorithm_1} with $\bar{\lambda} = \bar{\omega}^{\{j\}}$
               \If {Theorem~\ref{algorithm_1} is feasible} 
                 \State $\bar{\omega}^{\{j\}}=\bar{\lambda}^{\{m\}} + \theta^{\{m\}}\beta^{*\{m\}}v^{\{j\}}$,
                \Else
                 \State $\bar{\omega}^{\{j\}}=\bar{\lambda}^{\{m\}} + \theta^{\{m\}}\alpha^{\{m\}}v^{\{j\}}$,
               \EndIf
    \EndFor
        \State Set $\bar{\lambda}^{\{m+1\}}=\bar{\omega}^{\{j\}}$ that gives minimum $\beta^{*\{m,j\}}$
        \State Set $\beta^{*\{m+1\}} = \min_{j = 1,\, \ldots, \, n_v}\beta^{*\{m,j\}}$
        \State $m\leftarrow m+1$
    \EndWhile
    \State $R^{*} = 1/\sqrt{\beta^{*\{m\}}}$
    \end{algorithmic}
\label{algorithm_1}
\end{algorithm}

\subsection{Convexification of Theorem \ref{theorem:local_state_bounds}}
To pose Theorem \ref{theorem:local_state_bounds} as a convex optimisation, the upper bounds of the Lagrange multipliers $\bar{\lambda}$ have to be fixed. This restriction motivates the use of an iterative algorithm to refine the choice of $\bar{\lambda}$. In the following, an initialisation and update rule for $\bar{\lambda}$ is proposed which is then embedded within Algorithm \ref{algorithm_1} to iteratively generate new bounds $R^*$ and help reduce the conservatism of the approach.

\textbf{Initialisation of $ \bar{\lambda} $:} One way to initialise $\bar{\lambda}^{\{0\}}$ in Algorithm \ref{algorithm_1} is to first solve Theorem \ref{theorem:local_state_bounds} except with \eqref{eq:theorem_local_state_bound_constraint_6} replaced by 
\begin{align} \label{eq:lambda_mod}
    \begin{bmatrix}
    \lambda_k & \bar{\lambda}_k^{1/2} \\ \bar{\lambda}_k^{1/2}  & 1
    \end{bmatrix} \succeq 0, \quad \forall k = 1, \,\dots, \, K+n.
\end{align}
The above enforces $ {\lambda} \geq  \bar{\lambda}>0$ instead of the upper bound of \eqref{eq:theorem_local_state_bound_constraint_6}. The reason for replacing \eqref{eq:theorem_local_state_bound_constraint_6} with \eqref{eq:lambda_mod} in this modified version of Theorem \ref{theorem:local_state_bounds} is because $\bar{\lambda}^{1/2}$ can then be defined as a matrix variable to be searched over in Step 10 of the algorithm, giving flexibility. It is stressed though that this formulation of the problem can only be used to initialise $\bar{\lambda}^{\{0\}}$, as it does not generate stability certificates as \eqref{eq:theorem_local_state_bound_constraint_6}  would not hold.

\textbf{Update of $ \bar{\lambda} $:}
Step 10 in Algorithm \ref{algorithm_1} updates the upper bounds for the Lagrange multipliers $\bar{\lambda}^{\{0\}}$ to reduce the conservatism. In this step, candidate values for $\bar{\lambda}^{\{0\}}$ are proposed by stepping a distance $\theta \in \mathbb{R}$ in a direction  $v^{\{j\}}$, which are both defined before the inner loop on $j$. In this work, the directions $v$ were set to be all (normalised) combinations of the basis vectors of dimension $n+ K$ and their opposite directions, for instance
\begin{subequations}\begin{align}
    v^{\{1\}} &= \begin{bmatrix}   1, & 0, & \dots\,, & 0 \end{bmatrix}^\top, \\
     v^{\{2\}} &= \begin{bmatrix}   0, & 1, & \dots\,, & 0 \end{bmatrix}^\top, \\
     v^{\{3\}} &= \begin{bmatrix}   1, & 1, & \dots\,, & 0 \end{bmatrix}^\top/\sqrt{2}, \\
     v^{\{4\}} &=\begin{bmatrix}   -1, & -1, & \dots\,, & 0 \end{bmatrix}^\top/\sqrt{2}, 
\end{align}\end{subequations}
and so on for $j = 1\, \dots, \, n_v$, with the step lengths $\theta^{\{m\}} = 1$ and $\alpha^{\{m\}}=10^{4}$. The algorithm then  takes the value of $\bar{\lambda}^{\{m\}}$ which gave the biggest increase in $\beta^* = 1/R^2$, and then continues onto the next iterate.

\section{Numerical Results}
\label{sec:numerical_examples}

 \begin{table*}
\centering
\begin{scriptsize}
\begin{tabular}{|c|c|c|c|c|c|c|c|c|c|c|}
\hline 
 & \multicolumn{5}{c|}{\textbf{WKH model}}  & \multicolumn{5}{c|}{\textbf{9-state model}} \\ \hline
\textbf{Reynolds number ($Re$)} & \textbf{5} & \textbf{10} & \textbf{15} & \textbf{20} & \textbf{25}
 & \textbf{100} & \textbf{125} & \textbf{150} & \textbf{175} & \textbf{200}\\ \hline
\textbf{Methodology} & \multicolumn{10}{c|}{\textbf{Maximum energy perturbation ($R^*$)}} \\ \hline
{\color[HTML]{000000} Upper Limit (Simulation)} & 0.221 & 0.0480 & 0.0215 & 0.0124 & 0.0081  & // & // & // & // & // \\ \hline
Algorithm 1 & 0.120 & 0.0282 & 0.0121 & 0.0076 & 0.0049 & 0.0031 & 0.0019 & 0.0013 & 0.0011 & 0.0009  \\ \hline
Kalur, Mushtaq, Seiler \& Hemati & 0.100 & 0.0232 & 0.0100 & 0.0054 & 0.0034 &  0.0024 & 0.0015 & 0.0010 & 0.0007 & 0.0006  \\ \hline
Liu \& Gayme & 0.0383 & 0.0063 & 0.0021 & 0.0010 & 0.0006 & 0.0010 & 0.0006 & 0.0004 & 0.0003 & 0.0002 \\ \hline
\end{tabular}
\end{scriptsize}
\caption{Maximum energy perturbations $R^{*}$ of both the WKH and the 9-state Couette flow models.}
\label{table:results}
\end{table*}

\begin{figure}
      \centering
      \includegraphics[width=0.48\textwidth]{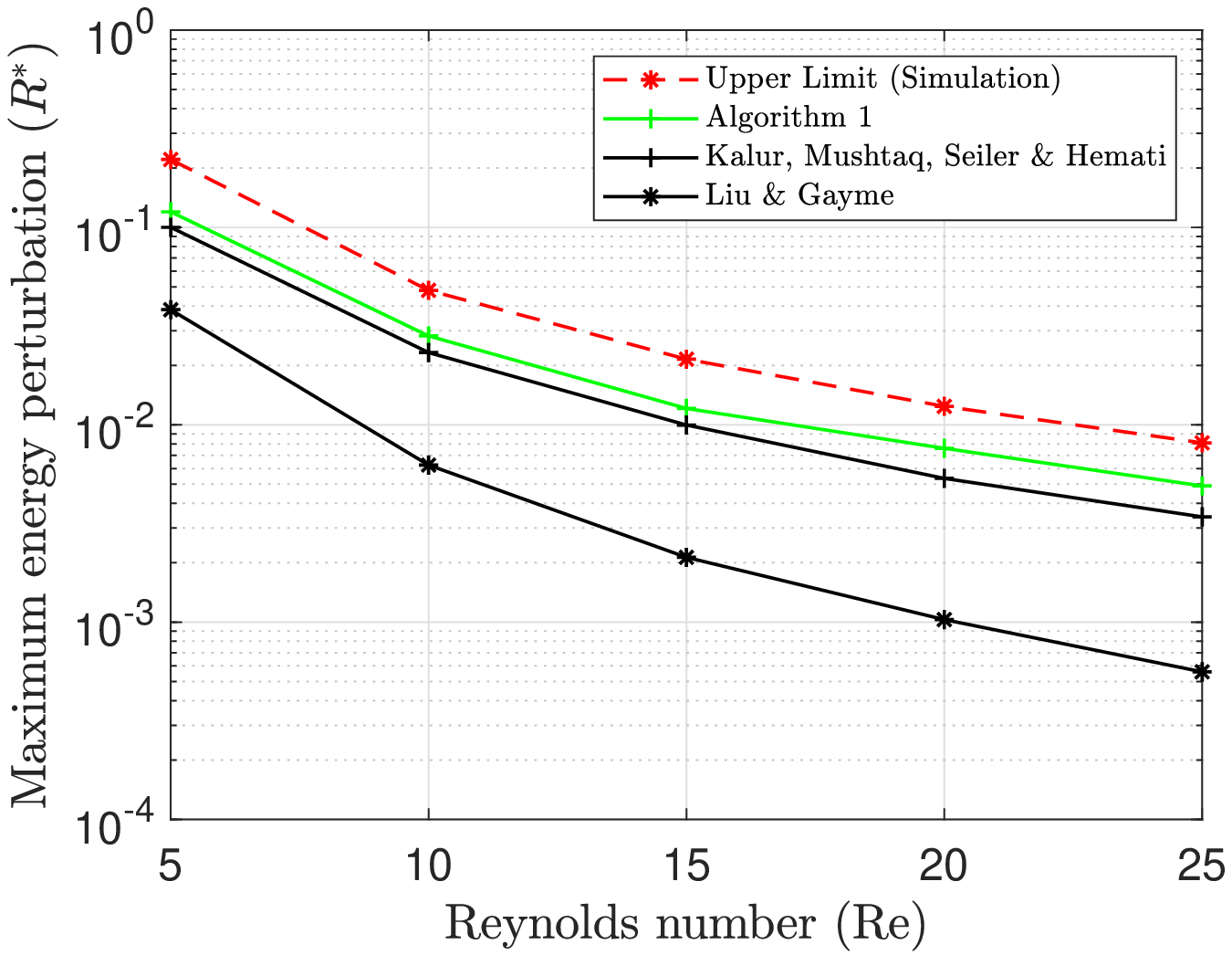}
      \caption{\textit{Comparison between \cite{Seiler_2}, \cite{Dennice} and Algorithm \ref{algorithm_1} for bounding the maximum energy perturbations $R^{*}$ of the WKH shear flow model.}}
      \label{fig:WKH_inner_estimates}
 \end{figure}
Numerical examples are now shown to evaluate the performance of the Algorithm \ref{algorithm_1} in computing inner estimates of the region of attraction for the four and nine state models. For both models, the maximal achievable energy perturbation obtained using Theorem \ref{theorem:local_state_bounds} was compared against \cite{Dennice} and \cite{Seiler_2} as well as an upper limit produced by simulating the system for different initial conditions $x_0$'s within a neighbourhood of its origin. For both examples, the toolbox CVX \cite{cvx} along with the solver MOSEK were used to generate the results presented below with tolerances $\epsilon=\epsilon_{\beta}=10^{-6}$.  The complete set of numerical results are detailed in Table \ref{table:results}.

\subsection{Waleffe-Kim-Hamilton (WKH) shear flow model}
 Figure \ref{fig:WKH_inner_estimates} shows the comparison for the WKH model described in Section~\ref{sec:transitional_flow_models}. For this example,  the WKH model's parameters were set to  $\sigma = \lambda  = \ldots= \delta= 1$. It is noted that other parameter values have been used for this model, notably in \cite{Waleffe}, but the choice of unity was selected to enable a direct comparison to the results of \cite{Seiler_1}. The maximal energy perturbation $R^{*}$ was found for Reynolds numbers (Re) in the range $[5:5:25]$.

 This figure shows the maximal energy perturbations allowed for both Algorithm \ref{algorithm_1} (green), \cite{Seiler_1}, \cite{Dennice} (black), and the upper limit found through system's simulations (red). The results presented in Table \ref{table:results} highlight how local quadratic bounds and the flexible computation of the ellipsoidal sets in Theorem \ref{theorem:local_state_bounds} have significantly reduced the conservatism, meaning that a higher energy perturbation is allowed. Specifically, the average improvement  over \cite{Seiler_1} for the five Reynolds numbers was $\approx 29\%$, while the average improvement over \cite{Dennice}  for the five Reynolds numbers was $\approx 482\%$.

  Figure~\ref{fig:convergence} compares the convergence rate of Algorithm \ref{algorithm_1} against \cite[Algorithm A]{Seiler_1} for the WKH model with a Reynolds number of Re $=10$, showing how Algorithm \ref{algorithm_1} required solving fewer optimisation problems to converge on its final value of $R^*$.

 \begin{figure}
       \centering
       \includegraphics[scale=0.53]{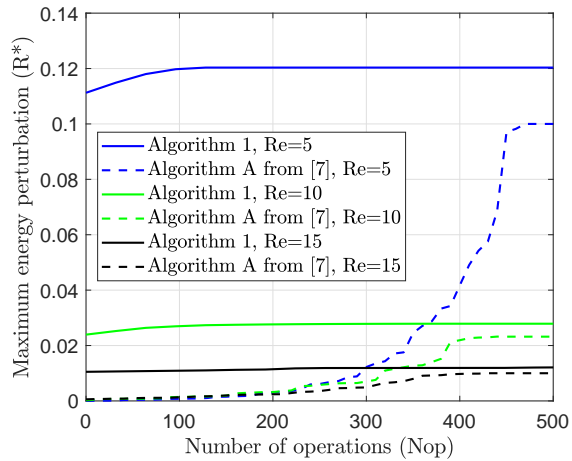}
      \caption{\textit{Comparison between the convergence rates of Algorithm \ref{algorithm_1} and \cite[Algorithm A]{Seiler_1} for bounding the maximum energy perturbations $R^{*}$ of the WKH model with different Reynolds numbers. $N_\text{op}$ corresponds to the number of optimisation problems solved to reach convergence.}}
       \label{fig:convergence}
  \end{figure}

\subsection{9-state reduced-order model of Couette flow}

A flow domain of $L_b=1.75\pi$ and $L_d=1.2\pi$ was defined for the numerical evaluation of the 9-state reduced-order model of Couette flow. Figure \ref{fig:9_state_inner_estimates} compares the maximum achievable energy perturbation $R^{*}$ for which stability could be verified, comparing Algorithm  \ref{algorithm_1} (green) against \cite{Seiler_1} and  \cite{Dennice} (black). Unlike for the WKH model, no upper limit could be found for this model from numerical simulations. The benefits of Theorem \ref{theorem:local_state_bounds} were more striking for this model compared against the WKH model, with the improvement averaged across the Reynolds' numbers Re $= [100:25:200]$ being $\approx 38\%$ over \cite{Seiler_1} and $\approx 253\%$ over \cite{Dennice}.

\begin{figure}[t]
      \centering
      \includegraphics[width=0.5\textwidth]{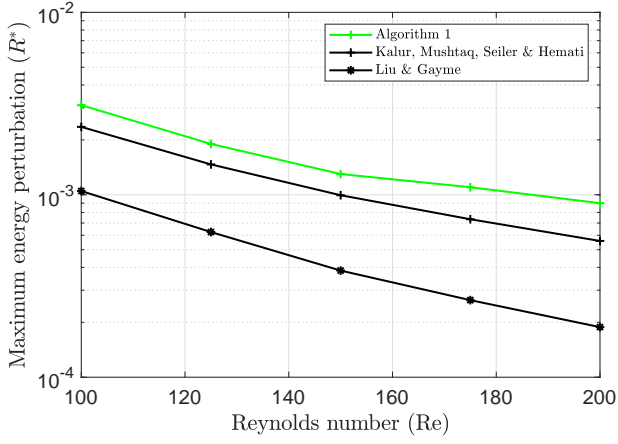}
     \caption{\textit{Comparison between \cite{Seiler_1}, \cite{Dennice} and Algorithm \ref{algorithm_1} for bounding the maximum energy perturbations $R^{*}$ of the 9-state Couette flow model.}}
      \label{fig:9_state_inner_estimates}
 \end{figure}

\section{Conclusions}
\label{sec:conclusion}

The regional stability analysis of transitional fluid flow models was considered. By exploiting the fact that the nonlinearities of these models are lossless and can be locally bounded, a method was proposed to determine the maximum energy perturbation in the flow field for which stability could be guaranteed. To illustrate the potential of the proposed method, numerical examples were demonstrated for both the WKH model and a 9-state model for Couette flow, showing reduced conservatism over the current state-of-the-art without sacrificing on computational efficiency. Future work will explore adapting the method for generic candidate Lyapunov function structures and developing more effective ways to convexify the problem than the proposed method involving $\bar{\lambda}.$

\section*{Acknowledgements}

The authors gratefully acknowledge EDF Energy, UK and the University of Oxford for supporting this research through a French internship scheme (sponsor license number UED4UGNF1). Ross Drummond was funded through a UKIC Fellowship from the Royal Academy of Engineering.

\bibliography{bibliography.bib}

\begin{thebibliography}{10}
\providecommand{\url}[1]{#1}
\csname url@samestyle\endcsname
\providecommand{\newblock}{\relax}
\providecommand{\bibinfo}[2]{#2}
\providecommand{\BIBentrySTDinterwordspacing}{\spaceskip=0pt\relax}
\providecommand{\BIBentryALTinterwordstretchfactor}{4}
\providecommand{\BIBentryALTinterwordspacing}{\spaceskip=\fontdimen2\font plus
\BIBentryALTinterwordstretchfactor\fontdimen3\font minus
  \fontdimen4\font\relax}
\providecommand{\BIBforeignlanguage}[2]{{%
\expandafter\ifx\csname l@#1\endcsname\relax
\typeout{** WARNING: IEEEtran.bst: No hyphenation pattern has been}%
\typeout{** loaded for the language `#1'. Using the pattern for}%
\typeout{** the default language instead.}%
\else
\language=\csname l@#1\endcsname
\fi
#2}}
\providecommand{\BIBdecl}{\relax}
\BIBdecl

\bibitem{Reynolds}
O.~Reynolds, ``{An experimental investigation of the circumstances which
  determine whether the motion of water in parallel channels shall be direct or
  sinuous and of the law of resistance in parallel channels},'' \emph{Philos.
  Trans. R. Soc.}, vol.~82, 1883.

\bibitem{Waleffe}
F.~Waleffe, ``Transition in shear flows. nonlinear normality versus nonnormal
  linearity,'' \emph{Physics of Fluids}, vol.~7, no.~12, p. {3060–3066},
  1995.

\bibitem{Waleffe_Fabian}
F.~Waleffe, J.~Kim, and J.~M. Hamilton, ``{On the Origin of Streaks in
  Turbulent Shear Flows},'' pp. 37--49, 1993.

\bibitem{Moehlis}
J.~Moehlis, H.~Faisst, and B.~Eckhardt, ``A low-dimensional model for turbulent
  shear flows,'' vol.~6, pp. 56--56, may 2004.

\bibitem{trefethen1}
J.~S. Baggett, T.~A. Driscoll, and L.~N. Trefethen, ``A mostly linear model of
  transition to turbulence,'' \emph{Physics of Fluids}, vol.~7, no.~4, pp.
  833--838, 1995.

\bibitem{trefethen2}
L.~N. Trefethen, A.~E. Trefethen, S.~C. Reddy, and T.~A. Driscoll,
  ``Hydrodynamic stability without eigenvalues,'' \emph{Science}, vol. 261, no.
  5121, pp. 578--584, 1993.

\bibitem{Seiler_1}
A.~Kalur, T.~Mushtaq, P.~Seiler, and M.~S. Hemati, ``{Estimating Regions of
  Attraction for Transitional Flows Using Quadratic Constraints},'' \emph{IEEE
  Control Systems Letters}, vol.~6, pp. 482--487, 2022.

\bibitem{Seiler_2}
A.~Kalur, P.~Seiler, and M.~S. Hemati, ``{Nonlinear stability analysis of
  transitional flows using quadratic constraints},'' \emph{Phys. Rev. Fluids},
  vol.~6, p. 044401, Apr 2021.

\bibitem{Dennice}
C.~Liu and D.~F. Gayme, ``{Input-output inspired method for permissible
  perturbation amplitude of transitional wall-bounded shear flows},''
  \emph{Phys. Rev. E}, vol. 102, p. 063108, Dec 2020.

\bibitem{Khalil}
H.~K. Khalil, ``{Nonlinear systems},'' \emph{3rd ed. Prentice-Hall}, 2002.

\bibitem{Kerswell}
R.~Kerswell, ``{Nonlinear Nonmodal Stability Theory},'' \emph{Annual Review of
  Fluid Mechanics}, vol.~50, no.~1, pp. 319--345, 2018.

\bibitem{Goulart}
P.~J. Goulart and S.~Chernyshenko, ``{Global stability analysis of fluid flows
  using sum-of-squares},'' \emph{Physica D}, vol. 241, pp. 692--704, 2012.

\bibitem{Valmorbida2}
G.~Valmorbida, R.~Drummond, and S.~R. Duncan, ``{Regional Analysis of
  Slope-Restricted Lurie Systems},'' \emph{IEEE Transactions on Automatic
  Control}, vol.~64, no.~3, pp. 1201--1208, 2019.

\bibitem{Fabian_Waleffe}
F.~Waleffe, ``On a self-sustaining process in shear flows,'' \emph{Physics of
  Fluids}, vol.~9, no.~4, pp. 883--900, 1997.

\bibitem{Holmes}
P.~Holmes, J.~L. Lumley, and G.~Berkooz, ``{Turbulence, Coherent Structures,
  Dynamical Systems and Symmetry},'' \emph{Cambridge University Press}, 1996.

\bibitem{Schmid}
P.~J. Schmid and D.~S. Henningson, ``{Stability and transition in shear
  flows},'' \emph{Springer}, 2001.

\bibitem{Leray}
J.~Leray, ``Sur le mouvement d'un liquide visqueux emplissant l'espace,''
  \emph{Acta Mathematica}, vol.~63, 1934.

\bibitem{Aizerman}
M.~A. Aizerman and F.~R. Gantmacher, ``Absolute stability of regulator
  systems,'' \emph{Holden-Day, CA}, 1964.

\bibitem{cvx}
M.~Grant and S.~Boyd, ``{CVX}: Matlab software for disciplined convex
  programming, version 2.1,'' \url{http://cvxr.com/cvx}, Mar. 2014.

\end{thebibliography}
\bibliographystyle{IEEEtran}

\section*{Appendix-9-state reduced-order model of Couette flow \cite{Holmes}}

\begin{subequations}
\begin{footnotesize}
\label{eq:9-state_model}
\allowdisplaybreaks
\begin{align}
    \dot{x}_1&=\frac{\beta^{2}}{R e}-\frac{\beta^{2}}{R e} x_{1}-\sqrt{\frac{3}{2}} \frac{\beta \gamma}{\kappa_{\alpha \beta \gamma}} x_{6} x_{8}+\sqrt{\frac{3}{2}} \frac{\beta \gamma}{\kappa_{\beta \gamma}} x_{2} x_{3},\\
    \dot{x}_2&=-\left(\frac{4 \beta^{2}}{3}+\gamma^{2}\right) \frac{x_{2}}{\operatorname{Re}}+\frac{5 \sqrt{2}}{3 \sqrt{3}} \frac{\gamma^{2}}{\kappa_{\alpha \gamma}} x_{4} x_{6}-\frac{\gamma^{2}}{\sqrt{6} \kappa_{\alpha \gamma}} x_{5} x_{7}\nonumber\\ &-\frac{\alpha \beta \gamma}{\sqrt{6} \kappa_{\alpha \gamma} \kappa_{\alpha \beta \gamma}} x_{5} x_{8}-\sqrt{\frac{3}{2}} \frac{\beta \gamma}{\kappa_{\beta \gamma}} x_{1} x_{3}-\sqrt{\frac{3}{2}} \frac{\beta \gamma}{\kappa_{\beta \gamma}} x_{3} x_{9},\\
    \dot{x}_3&=-\frac{\beta^{2}+\gamma^{2}}{\operatorname{Re}} x_{3}+\frac{2}{\sqrt{6}} \frac{\alpha \beta \gamma}{\kappa_{\alpha \gamma} \kappa_{\beta \gamma}}\left(x_{4} x_{7}+x_{5} x_{6}\right)\nonumber\\
        &+\frac{\beta^{2}\left(3 \alpha^{2}+\gamma^{2}\right)-3 \gamma^{2}\left(\alpha^{2}+\gamma^{2}\right)}{\sqrt{6} \kappa_{\alpha \gamma} \kappa_{\beta \gamma} \kappa_{\alpha \beta \gamma}} x_{4} x_{8}, \\
    \dot{x}_4&=-\frac{3 \alpha^{2}+4 \beta^{2}}{3 R e} x_{4}-\frac{\alpha}{\sqrt{6}} x_{1} x_{5}-\frac{10}{3 \sqrt{6}} \frac{\alpha^{2}}{\kappa_{\alpha \gamma}} x_{2} x_{6}\nonumber\\
          &-\sqrt{\frac{3}{2}} \frac{\alpha \beta \gamma}{\kappa_{\alpha \gamma} \kappa_{\beta \gamma}} x_{3} x_{7}-\sqrt{\frac{3}{2}} \frac{\alpha^{2} \beta^{2}}{\kappa_{\alpha \gamma} \kappa_{\beta \gamma} \kappa_{\alpha \beta \gamma}} x_{3} x_{8}-\frac{\alpha}{\sqrt{6}} x_{5} x_{9},\\
\allowdisplaybreaks
 \dot{x}_5&=-\frac{\alpha^{2}+\beta^{2}}{\operatorname{Re}} x_{5}+\frac{\alpha}{\sqrt{6}} x_{1} x_{4}+\frac{\alpha^{2}}{\sqrt{6} \kappa_{\alpha \gamma}} x_{2} x_{7}\nonumber\\ 
           &-\frac{\alpha \beta \gamma}{\sqrt{6} \kappa_{\alpha \gamma} \kappa_{\alpha \beta \gamma}} x_{2} x_{8}+\frac{\alpha}{\sqrt{6}} x_{4} x_{9}+\frac{2}{\sqrt{6}} \frac{\alpha \beta \gamma}{\kappa_{\alpha \gamma} \kappa_{\beta \gamma}} x_{3} x_{6},\\
  \dot{x}_6&=-\frac{3 \alpha^{2}+4 \beta^{2}+3 \gamma^{2}}{3 R e} x_{6}+\frac{\alpha}{\sqrt{6}} x_{1} x_{7}+\sqrt{\frac{3}{2}} \frac{\beta \gamma}{\kappa_{\alpha \beta \gamma}} x_{1} x_{8},\\
    \dot{x}_7&=-\frac{\alpha^{2}+\beta^{2}+\gamma^{2}}{R e} x_{7}-\frac{\alpha}{\sqrt{6}}\left(x_{1} x_{6}+x_{6} x_{9}\right)\nonumber\\
            &+\frac{1}{\sqrt{6}} \frac{\gamma^{2}-\alpha^{2}}{\kappa_{\alpha \gamma}} x_{2} x_{5}+\frac{1}{\sqrt{6}} \frac{\alpha \beta \gamma}{\kappa_{\alpha \gamma} \kappa_{\beta \gamma}} x_{3} x_{4},\\
  \dot{x}_8&=-\frac{\alpha^{2}+\beta^{2}+\gamma^{2}}{R e} x_{8}+\frac{2}{\sqrt{6}} \frac{\alpha \beta \gamma}{\kappa_{\alpha \gamma} \kappa_{\alpha \beta \gamma}} x_{2} x_{5}\nonumber\\
            &+\frac{\gamma^{2}\left(3 \alpha^{2}-\beta^{2}+3 \gamma^{2}\right)}{\sqrt{6} \kappa_{\alpha \gamma} \kappa_{\beta \gamma} \kappa_{\alpha \beta \gamma}} x_{3} x_{4},\\
   \dot{x}_9&=-\frac{9 \beta^{2}}{\operatorname{Re}} x_{9}+\sqrt{\frac{3}{2}} \frac{\beta \gamma}{\kappa_{\beta \gamma}} x_{2} a_{3}-\sqrt{\frac{3}{2}} \frac{\beta \gamma}{\kappa_{\alpha \beta \gamma}} x_{6} x_{8}. 
\end{align}
\end{footnotesize}
\end{subequations}
\noindent where $\alpha=\frac{2\pi}{L_x}$, $\beta=\frac{\pi}{2}$, $\gamma=\frac{2\pi}{Lz}$, $\kappa_{\alpha \gamma}=\sqrt{\alpha^{2}+\gamma^{2}}$, $\kappa_{\beta \gamma}=\sqrt{\beta^{2}+\gamma^{2}}$ and $\kappa_{\alpha \beta \gamma}=\sqrt{\alpha^{2}+\beta^{2}+\gamma^{2}}$.

\end{document}